\documentclass[10pt]{article}
\usepackage{amsmath,amsthm,amsfonts,amssymb,latexsym,graphicx}

\usepackage{algorithm}
\usepackage[noend]{algpseudocode}
\algrenewcommand\algorithmicthen{\relax}
\algdef{C}[IF]{IF}{ElsIf}[1]{\textbf{elif}\ #1}
\algrenewcommand\algorithmicdo{\relax}

\usepackage[colorlinks=true,citecolor=blue,pdfpagemode=UseNone,pdfstartview=FitH]{hyperref}

\newcommand*{\eII}{Vovk/Wang:2021}
\newcommand*{\eIII}{Vovk/Wang:2023}
\newcommand*{\eIV}{Vovk/Wang:arXiv2003}
\newcommand*{\eV}{Vovk/Wang:2024EJS}

\allowdisplaybreaks[4]

\renewcommand{\d}{\,\mathrm{d}}

\newcommand{\eP}{\lozenge}
\newcommand{\eN}{\square}

\newcommand{\E}{\mathbb{E}}
\renewcommand{\P}{\mathbb{P}}
\newcommand{\R}{\mathbb{R}}
\newcommand{\N}{\mathbb{N}}

\newcommand{\FFF}{\mathcal{F}}

\newcommand{\QQQ}{\mathcal{Q}}

\renewcommand{\complement}{\textsf{c}}

\theoremstyle{plain}
\newtheorem{theorem}{Theorem}

\newtheorem{lemma}[theorem]{Lemma}
\newtheorem{proposition}[theorem]{Proposition}

\theoremstyle{definition}

\theoremstyle{remark}

\newlength{\IndentI}
\newlength{\IndentII}
\newlength{\IndentIII}
\newlength{\IndentIV}
\newlength{\IndentV}
\title{Multiple testing in game-theoretic probability:
  pictures and questions}
\author{Vladimir Vovk}

\begin{document}
\maketitle

  \begin{abstract}
    The usual way of testing probability forecasts in game-theoretic probability
    is via construction of test martingales.
    The standard assumption is that all forecasts are output by the same forecaster.
    In this paper I will discuss possible extensions of this picture
    to testing probability forecasts output by several forecasters.
    This corresponds to multiple hypothesis testing in statistics.
    One interesting phenomenon is that even a slight relaxation of the requirement of family-wise validity
    leads to a very significant increase in the efficiency of testing procedures.
    The main goal of this paper is to report results of preliminary simulation studies
    and list some directions of further research.

    \medskip

    \noindent
    The version of this paper at \url{http://alrw.net/e} (Working Paper 10)
    is updated most often.
  \end{abstract}

\section{Introduction}

Game-theoretic probability,
as presented in, e.g.,
my joint books \cite{Shafer/Vovk:2001} and \cite{Shafer/Vovk:2019} with Glenn Shafer,
is based on the idea that a null hypothesis can be tested dynamically
by gambling against it.
More generally, we are testing a player called Forecaster,
which can be a scientific theory, a computer program, a human forecaster, etc.
The gambler starts from an initial capital of 1 and is required to keep his capital nonnegative.
His current capital is interpreted as the degree to which the null hypothesis has been undermined.

The idea of testing via gambling goes back
at least to Richard von Mises's principle of the impossibility of a gambling system
(Unm\"oglichkeit eines Spielsystems \cite[p.~58]{Mises:1919}),
but von Mises's notion of gambling was too narrow,
and it was only applicable to infinite sequences.
The narrowness of von Mises's notion of gambling
was demonstrated by Ville \cite[Sect.~II.4]{Ville:1939}
(for an English translation, see \cite{Shafer:2005}).
Ville proposed extending von Mises's testing procedure
to using nonnegative martingales \cite[Chap.~IV]{Ville:1939},
but somewhat surprisingly, did not explicitly apply his wider notion of testing
to restate von Mises's principle of the impossibility of a gambling system.
It appears that the idea of testing using nonnegative martingales
emerged gradually in various fields,
including the algorithmic theory of randomness.

In this paper we will be interested in testing several forecasters in one go,
with different forecasters being tested at different steps.
Testing by gambling can be studied in the usual setting of measure-theoretic probability,
and this is what we will do in this paper,
for simplicity and as a first step.
Replacing measure-theoretic probability by game-theoretic probability
as mathematical foundation for our definitions and results
will be one of directions of future research.
For now, each forecaster will be formalized as a composite null hypothesis,
represented by a set of probability measures on the sample space.

In principle, we can consider two settings for testing multiple null hypotheses.
In the \emph{closed} setting,
we have a fixed number $K$ of null hypotheses.
In the \emph{open} setting,
the number of null hypotheses is not known in advance
and is potentially infinite.
In this paper we will concentrate on the closed setting.

This paper has been prepared in support of my planned talk
at the Oberwolfach workshop ``Game-theoretic statistical inference:
optional sampling, universal inference, and multiple testing based on e-values''
organized by Peter Gr\"unwald, Aaditya Ramdas, Ruodu Wang, and Johanna Ziegel
(5--10 May 2024).

\section{Dynamic necessity and possibility}
\label{sec:necessity}

Let us fix a probability space $(\Omega,\FFF,\P)$
equipped with a filtration $\FFF=(\FFF_n)_{n=0}^{\infty}$,
so that $\FFF_0\subseteq\FFF_1\subseteq\dots\subseteq\FFF$
is a nested sequence of $\sigma$-algebras.
Apart from the true probability measure $\P$
we will often consider other probability measures $Q$
on the measurable space $(\Omega,\FFF)$.
Our notation for the expectation of a random variable
$f:\Omega\to[0,\infty]$ w.r.\ to $Q$ will be $\E_Q(f):=\int f \d Q$,
abbreviated to $\E(f)$ when $Q=\P$
(in general, ``w.r.\ to $Q$'' or the indication of $Q$ is usually omitted
when $Q=\P$).
Let $\QQQ$ be the family of all probability measures on $(\Omega,\FFF)$.

A \emph{test martingale} $S$ w.r.\ to $Q$ is a sequence $S_0,S_1,\dots$
of random variables taking values in $[0,\infty]$
such that $S_0=1$ and $\E_Q(S_n\mid\FFF_{n-1})=S_{n-1}$ for all $n=1,2,\dots$.
A \emph{martingale test} is a family $(S^Q)_{Q\in\QQQ}$ of test martingales $S^Q$
w.r.\ to $Q$.
At each time $n$,
we interpret $S^Q_n(\omega)$ as a measure of disagreement
between the realized outcome $\omega$ and its putative explanation $Q$;
we may say that $\omega$ is \emph{$\alpha$-strange} at time $n$ w.r.\ to $Q$
if $S^Q_n(\omega)\ge\alpha$.

Fix a martingale test $(S^Q)$
and let $A\subseteq\QQQ$ be a property of a probability measure $Q$
(with the property being satisfied if and only if $Q\in A$).
The \emph{necessity measure} of $A$ at time $n$
in view of the realized outcome $\omega\in\Omega$ is
\[
  \eN_n(A\mid\omega)
  :=
  \inf_{Q:Q\notin A}
  S^Q_n(\omega),
\]
and the \emph{possibility measure} of $A$ in view of $\omega$ is
\begin{equation*}
  \eP(A\mid\omega)
  :=
  \inf_{Q:Q\in A}
  S^Q_n(\omega)
  =
  \eN_n(A^{\complement}\mid\omega).
\end{equation*}
The interpretation is that $A$ holds
unless $\omega$ is $\eN_n(A\mid\omega)$-strange at time $n$,
and similarly for $\eP$.
It is important that this property of validity
can be applied to all $A$ at the same time;
the martingale test, however, should be chosen in advance.

\section{Multiple testing of a single null hypothesis}

In this section, we fix a probability measure $Q\in\QQQ$ on the sample space;
we are interested in testing whether $Q$ is the true probability measure, $Q=\P$.
For that, we would like to have one test martingale w.r.\ to $Q$.

Instead, we are given $K$ test martingales $S^{(k)}$ for $k=1,\dots,K$.
In the language of game-theoretic probability as presented in \cite{Shafer/Vovk:2019},
we have $K$ Sceptics testing $Q$ as null hypothesis.
Suppose the test martingales $S^{(1)},\dots,S^{(K)}$ are \emph{uncorrelated},
meaning that there exists a predictable sequence $k_n$, $n=1,2,\dots$,
such that $S^{(k)}_n=S^{(k)}_{n-1}$ for all $n$ and all $k\ne k_n$.
(And the requirement of predictability means that each $k_n$ is $\FFF_{n-1}$-measurable.)
This concept and terminology goes back to Shafer \cite[Sect.~12.3]{Shafer:1996}
(at least for the case $K=2$).
The interpretation is that Forecaster is being tested by Sceptic $k_n$ on step $n$.

A convex combination of test martingales is always a test martingale.
In this section we discuss how else we can combine test martingales.
First we notice that the product $S^{(1)}\dots S^{(K)}$
(as well as the product of a subset of $S^{(1)},\dots,S^{(K)}$)
is a test martingale \cite[Proposition 12.5(1)]{Shafer:1996}.
Indeed, dropping the lower index $Q$,
\begin{align*}
  \E
  \left(
    S^{(1)}_n \dots S^{(K)}_n
    \mid
    \FFF_{n-1}
  \right)
  &=
  \E
  \left(
    S^{(1)}_{n-1} \dots S^{(k_n-1)}_{n-1}
    S^{(k_n)}_{n}
    S^{(k_n+1)}_{n-1} \dots S^{(K)}_{n-1}
    \mid
    \FFF_{n-1}
  \right)\\
  &=
  S^{(1)}_{n-1} \dots S^{(k_n-1)}_{n-1}
  S^{(k_n+1)}_{n-1} \dots S^{(K)}_{n-1}
  \E
  \left(
    S^{(k_n)}_{n}
    \mid
    \FFF_{n-1}
  \right)\\
  &=
  S^{(1)}_{n-1} \dots S^{(K)}_{n-1}.
\end{align*}

A \emph{martingale merging function} is a measurable function
$F:[0,\infty)^K\to[0,\infty)$
such that $F(S^{(1)}_n,\dots,S^{(K)}_n)$, $n=0,1,\dots$, is a test martingale
whenever $S^{(1)},\dots,S^{(K)}$ are test martingales
(and we require this to hold for any probability space and any test martingales on it).
We will apply such functions to base test martingales
to get a new test martingale that can be used for testing.
Our definition allows test martingales to take value $\infty$,
and so we extend each martingale merging function in a canonical way
(as in \cite[Sect.~3]{Vovk/Wang:2021}):
namely, we set $F:=\infty$ whenever one or more of its arguments are $\infty$.

An example of a martingale merging function is
(cf.\ \cite{\eII,\eIV})
\begin{equation*}
  \begin{aligned}
    U_n(s_1,\dots,s_K)
    &:=
    \frac{1}{\binom{K}{n}}
    \sum_{\{k_1,\dots,k_n\}\subseteq\{1,\dots,K\}}
    s_{k_1} \dots s_{k_n}\\
    &=
    \frac{1}{\binom{K}{n}}
    \sigma_n(s_1,\dots,s_K),
    \quad
    n\in\{1,\dots,K\},
  \end{aligned}
\end{equation*}
where $\sigma_n$ is the \emph{$n$th elementary symmetric polynomial in $K$ variables}.
In other words, $U_n$ is $\sigma_n$ normalized by dividing by $\sigma(1,\dots,1)$;
normalization ensures that the initial value of the combination of test martingales
starts from 1 as initial capital (and then it is a test martingale).
We will be particularly interested in the cases $n=1$ and $n=2$.

In my previous joint papers with Ruodu Wang \cite{\eII,\eIV},
we referred to the functions $U_n$ as U-statistics,
but this is potentially confusing as we are omitting
``with product as kernel''
as far as the standard statistical notion of U-statistics is concerned.
In this paper I will call $U_n$
\emph{normalized elementary symmetric polynomials} (NESPs).

A \emph{multiaffine polynomial} is defined as a multivariate polynomial
such that none of its monomials has any variable raised to power 2 or more.
(The less formal version ``multilinear polynomial'' of this term
is more popular in literature,
but would have been awkward in this paper.)
A multiaffine polynomial is \emph{positive} if each of its (non-zero) coefficients is positive.
It is \emph{normalized} if its value is 1 when all its arguments are 1.

\begin{proposition}\label{prop:main}
  For a fixed number of arguments $K$,
  every martingale merging function is a multiaffine polynomial
  that is positive and, of course, normalized.
\end{proposition}

Let us say that a function of several variables is \emph{symmetric}
if it is invariant w.r.\ to the permutations of its arguments.
Specializing Proposition~\ref{prop:main} to symmetric functions,
we obtain the following statement.

\begin{proposition}\label{prop:symmetric}
  For a fixed number of arguments $K$,
  every symmetric martingale merging function
  is a convex mixture of the NESPs $U_n$, $n=0,\dots,K$.
\end{proposition}

In Proposition~\ref{prop:symmetric}, $U_0$ is understood to be $1$.
For proofs of Propositions~\ref{prop:main} and \ref{prop:symmetric},
see Appendix~\ref{app:proofs}.
From now on we will consider symmetric martingale merging functions.

\section{Family-wise multiple testing}
\label{sec:diagonal}

We are given $K$ adapted sequences $S^{(k)}=(S^{(k)}_1,S^{(k)}_2,\dots)$,
$k=1,\dots,K$,
of random variables taking values in $[0,\infty]$
and a predictable sequence $k_1,k_2,\dots$ of random variables taking values in $\{1,\dots,K\}$
such that $S^{(k)}_n=S^{(k)}_{n-1}$ whenever $k_n\ne k$.
Let us say that $k\in\{1,\dots,K\}$ is an \emph{anomalous index} for $Q\in\QQQ$
if $S^{(k)}$ is not a test martingale w.r.\ to $Q$
(with $S^{(k)}_0$ understood to be 1).

The interpretation is that at each step $n$ we are testing a null hypothesis $H_k\subseteq\QQQ$,
which leads to a change in $S^{(k)}_n$.
There are $K$ null hypotheses $H_1,\dots,H_K$,
and at step $n$ we are testing $H_{k_n}$.
The process of gambling is fair, in the sense of leading to a test martingale $S^{(k)}$,
under each $Q\in H_k$.
However, it does not have to be a test martingale
under the true probability  measure $\P$.
(A more realistic picture arises when we replace ``test martingale'' by ``e-process'',
i.e., a process dominated by a test martingale,
but let us concentrate on the simpler case of test martingales
in this paper.)

After observing the values of $S^{(k)}$ over steps $1,\dots,n$,
we might come up with a \emph{rejection set} $R\subseteq\{1,\dots,K\}$
containing the indices of the hypotheses that we decide to reject at step $n$.
It is natural to include in $R$ the indices $k$ with the largest values of $S^{(k)}_n$.
The elements of $R$ are \emph{discoveries}.
A discovery $k\in R$ is a \emph{true discovery} if $\P\notin H_k$,
and it is a \emph{false discovery} if $\P\in H_k$.
Let us also say that $k\in R$ is a \emph{justified discovery} if $k$ is an anomalous index.
Every justified discovery is a true discovery.
(The notion of a justified discovery is simpler than that of a true discovery
in that it does not involve the null hypotheses $H_k$.)

In this section we are interested in the necessity of all $k\in R$ being justified discoveries.
This number is a lower bound on the necessity of all $k\in R$ being true discoveries.
In other words, we are interested in conclusions that are family-wise valid.

For each $Q\in\QQQ$, let
\[
  J_Q
  :=
  \left\{
    k\in\{1,\dots,K\}
    \mid
    \text{$S^{(k)}$ is a test martingale under $Q$}
  \right\}.
\]
We are interested in the necessity of the property
\begin{equation}\label{eq:R-strict}
  R
  \cap
  J_Q
  =
  \emptyset
\end{equation}
that all discoveries in $R$ are justified.

The most natural martingale test in our current context
is obtained by applying a martingale merging function $F$
to the test martingales among the $S^{(k)}$.
In other words, for a given $Q$,
$F$ should be applied to $S^{(k)}$ for $k\in J_Q$.
Let us fix $F$.
Notice that we need $F$ for any number of arguments from 1 to $K$,
so formally we need a family $(F_k)_{k=1}^K$ of martingale merging functions.
We abbreviate $F_k(\dots)$ to $F(\dots)$
since $k$ is determined by the number of arguments and so redundant.

The optimal discovery sets at time $n$
are $R_{r,n}$, $r=1,\dots,K$,
where each $R_{r,n}\subseteq\{1,\dots,K\}$ has size $r$
and consists of the indices of the $r$ largest values in the set of $S^{(k)}_n$, $k=1,\dots,K$;
in the case of ties, let us give preference to smaller~$k$.
Define the \emph{chronological discovery diagonal} by
\begin{equation}\label{eq:CDD}
  \begin{aligned}
    \eN_n(R_{r,n}\cap J_Q=\emptyset)
    &=
    \inf_{Q\in\QQQ:R_{r,n}\cap J_Q\ne\emptyset}
    F
    \left(
      \left(
        S^{(k)}_n
      \right)_{k\in J_Q}
    \right)\\
    &\ge
    \inf_{I\subseteq\{1,\dots,K\}:R_{r,n}\cap I\ne\emptyset}
    F
    \left(
      \left(
        S^{(i)}_n
      \right)_{i\in I}
    \right)
    =:
    d_{r,n}.
  \end{aligned}
\end{equation}
(I will explain the origin of the term ``diagonal'' in Sect.~\ref{sec:DM}.)

\begin{algorithm}[bt]
  \caption{Chronological discovery diagonal $d_{r,n}$}
  \label{alg:CDD}
  \begin{algorithmic}[1]
    \Require
      symmetric martingale merging functions $F_k$, $k\in\{1,\dots,K\}$.
    \Require
      decreasing sequence of martingale values $S^1\ge\dots\ge S^K$.
    \For{$n=1,2,\dots$}
      \For{$r=1,\dots,K$}
        \State $d_{r,n}:=F((S^r))$
        \For{$k=r+1,\dots,K$}
          \State $S := F((S^i)_{i\in\{r\}\cup\{k,\dots,K\}})$
          \If{$S < d_{r,n}$}
	    \State $d_{r,n} := S$
          \EndIf
        \EndFor
      \EndFor
    \EndFor
  \end{algorithmic}
\end{algorithm}

Algorithm~\ref{eq:CDD} spells out the computation
of the infinite $K\times\infty$ matrix $d_{r,n}$,
although in our simulation studies we will only plot
paths $n\mapsto d_{r,n}$ for a few fixed $r$.
The algorithm assumes that the final martingale values
$S^{(k)}_n$ are sorted in the descending order,
and the sorted values are denoted $S^1\ge\dots\ge S^K$.
One of its inputs is a family of martingale merging functions $F_k$,
but as before, $F_k(\dots)$ is abbreviated to $F(\dots)$.

\begin{figure}
  \begin{center}
    \includegraphics[width=0.49\textwidth]{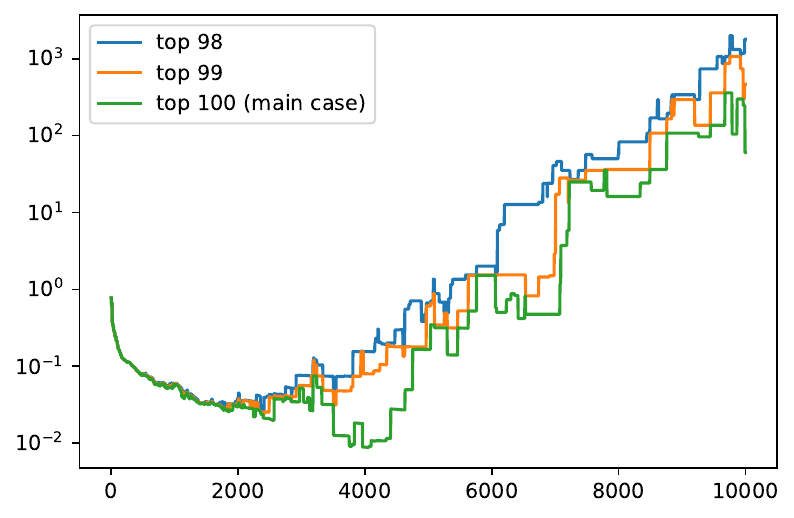}
    \includegraphics[width=0.49\textwidth]{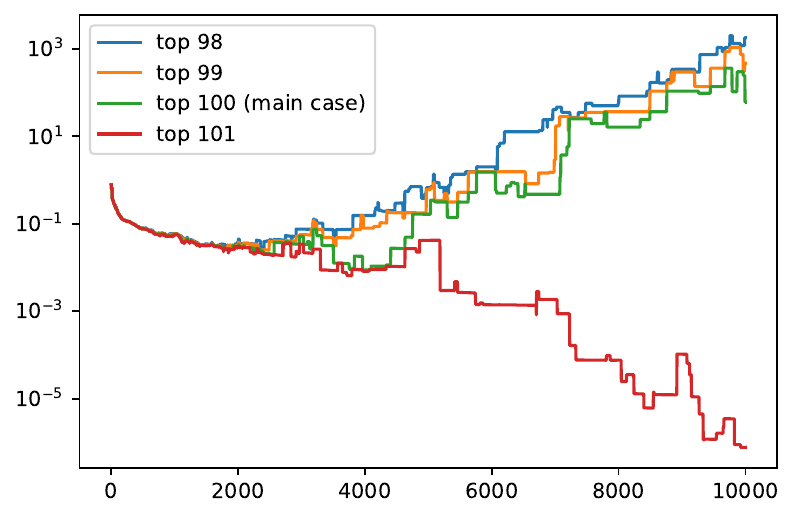}
  \end{center}
  \caption{Discovery plots for 100, 99, and 98 hypotheses.
    The right panel also adds the case of 101 hypotheses
    (at least one of which is bound to be wrong).}
  \label{fig:plot_diagonal}
\end{figure}

In our simulation studies we have 200 null hypotheses, all of them being $N(0,1)$,
numbered from 1 to 200.
The first 100 null hypotheses are false, and the true distribution is $N(-1,1)$;
and the remaining 100 null hypotheses are true.
At each step, from 1 to 10\,000,
we choose the hypothesis being tested randomly with equal probabilities,
so that each hypothesis is chosen with probability $0.5\%$.
Figure~\ref{fig:plot_diagonal} gives the plots $n\mapsto d_{r,n}$
for $r:=100$ (meaning that we aim to discover all 100 false null hypotheses),
$r:=99$, and $r:=98$.
Let us call such plots \emph{discovery plots}.
We generate the 10\,000 observations randomly
with the standard seed of 42 for the random number generator
(in fact, the results are very sensitive to the chosen value for the seed).
The final value $d_{100,10000}$ of the discovery plot
for the top 100 martingale values ($r:=100$) is approximately 60.1.
Using Jeffreys's \cite[Appendix B]{Jeffreys:1961} expression,
there is very strong evidence that the top 100 martingale values
exactly pinpoint the 100 false null hypotheses.

The martingale merging function used in Fig.~\ref{fig:plot_diagonal} is $U_1$.
It is clear that any symmetric martingale merging function,
which is a convex mixture of $U_n$ (Proposition~\ref{prop:main} above),
that does not have $U_1$ as its component,
will produce very poor results for all discovery plots
shown in Fig.~\ref{fig:plot_diagonal}:
e.g., $U_2(S^{100},S^{200})$ will be very small
(approximately $7.80\times10^{-25}$ in our case),
and $U_2(S^{100},S^k,\dots,S^{200})$, $k=101,\dots,199$, will be even smaller.

While Fig.~\ref{fig:plot_diagonal} uses the $U_1$ martingale merging function,
using, e.g., $(U_1+U_2)/2$ would give similar results.

\section{Almost family-wise multiple testing}

Let us now relax the requirement \eqref{eq:R-strict} to
\[
  \left|
    R
    \cap
    J_Q
  \right|
  \le
  1.
\]
This requirement can be interpreted as almost all discoveries in $R$ being justified:
we are allowing only one exception.
The chronological discovery diagonal~\eqref{eq:CDD}
now becomes the \emph{chronological discovery subdiagonal}
\begin{equation*}
  \begin{aligned}
    \eN_n(\left|R_{r,n}\cap J_Q\right|\le1)
    &=
    \inf_{Q\in\QQQ:\left|R_{r,n}\cap J_Q\right|>1}
    F
    \left(
      \left(
        S^{(k)}_n
      \right)_{k\in J_Q}
    \right)\\
    &\ge
    \inf_{I\subseteq\{1,\dots,K\}:\left|R_{r,n}\cap I\right|>1}
    F
    \left(
      \left(
        S^{(i)}_n
      \right)_{i\in I}
    \right)
    =:
    d'_{r,n}.
  \end{aligned}
\end{equation*}

\begin{algorithm}[bt]
  \caption{Chronological discovery subdiagonal $d'_{r,n}$}
  \label{alg:CDS}
  \begin{algorithmic}[1]
    \Require
      symmetric martingale merging functions $F_k$, $k\in\{1,\dots,K\}$.
    \Require
      decreasing sequence of martingale values $S^1\ge\dots\ge S^K$.
    \For{$n=1,2,\dots$}
      \For{$r=1,\dots,K$}
        \If{$r=1$}
          \State $I_r:=\{r\}$
        \Else
          \State $I_r:=\{r-1,r\}$
        \EndIf
          \State $d'_{r,n}:=F(I_r)$
        \For{$k=r+1,\dots,K$}
          \State $S := F((S^i)_{i\in I_r\cup\{k,\dots,K\}})$
          \If{$S < d'_{r,n}$}
	    \State $d'_{r,n} := S$
          \EndIf
        \EndFor
      \EndFor
    \EndFor
  \end{algorithmic}
\end{algorithm}

The analogue of Algorithm~\ref{alg:CDD} for the chronological discovery subdiagonal
is given as Algorithm~\ref{alg:CDS}.
In our simulation study we apply it to the martingale merging function $U_2$.
One complication is that it sometimes has to be applied to sequences of length 1,
in which case we understand it to be the same as $U_1$.

\begin{figure}
  \begin{center}
    \includegraphics[width=0.7\textwidth]{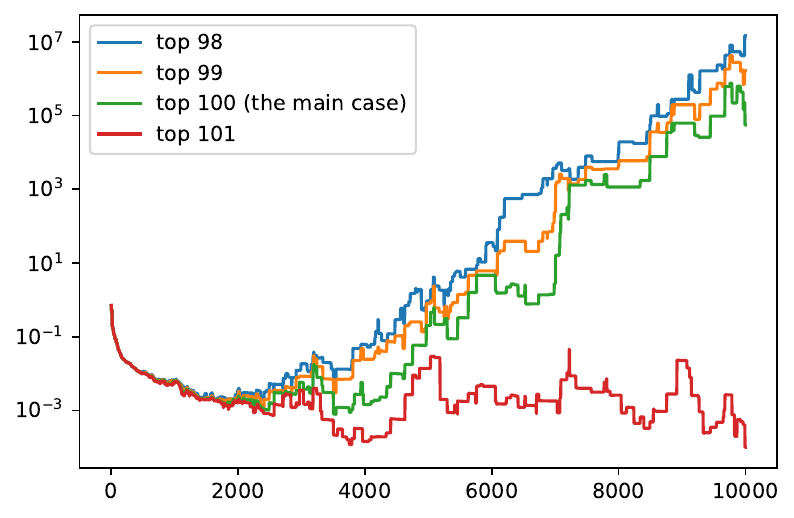}
  \end{center}
  \caption{Discovery plots for 101, 100, 99, and 98 hypotheses when one exception is allowed
    (almost family-wise case).}
  \label{fig:plot_subdiagonal}
\end{figure}

Figure~\ref{fig:plot_subdiagonal} is analogous to Fig.~\ref{fig:plot_diagonal}
but allows one exception and uses $U_2$ as martingale merging function.
In this case using $U_2$ works much better than $U_1$.
The final value of the discovery plot
for the top 100 martingale values ($r=100$) is approximately $1.07\times10^8$.

\section{Dynamic confidence regions}

Necessity measures discussed in Sect.~\ref{sec:necessity}
are only one possible way to package the idea of necessity.
A much more standard way is to use confidence regions,
which we will define in this section in our dynamic setting;
our definitions will be natural modifications of the standard definitions
(in the case of p-values)
and the definitions given in \cite{\eII,\eIV}
(in the case of e-values).
Let $(S^Q)$ be a martingale test,
fixed throughout this section.

We will be interested in confidence regions for the \emph{parameter} $g(\P)$,
where $g:\QQQ\to\Theta$ is a mapping from the probability measures on the sample space
to the parameter space $\Theta$ (which can be any set).
The \emph{exact confidence region} for $g(\P)$ at time $n$
corresponding to the realized outcome $\omega$
and significance level $\alpha>0$ is defined as
\[
  \Gamma^g_{\alpha,n}(\omega)
  :=
  \{g(Q)\mid S_n^Q(\omega)<\alpha\};
\]
as usual, the dependence on $\omega$ is often suppressed.
A \emph{confidence region} is a set of parameter values
containing the exact confidence region.

Alternatively, we can define a confidence region
as a set $A\subseteq\Theta$ such that
\begin{equation}\label{eq:inequality}
  \eN_n(\{Q\mid g(Q)\in A\})\ge\alpha.
\end{equation}
The exact confidence region $\Gamma^g_{\alpha,n}$ is the smallest such set.
In other words, $A:=\Gamma^g_{\alpha,n}$ satisfies \eqref{eq:inequality},
and any $A$ satisfying \eqref{eq:inequality} contains $\Gamma^g_{\alpha,n}$,
$A\supseteq\Gamma^g_{\alpha,n}$.

Finally, we can define the exact confidence region $\Gamma^g_{\alpha,n}$
as the set of all $\theta\in\Theta$ satisfying $\eP_n(g^{-1}(\theta))<\alpha$.

One disadvantage of the dynamic notion of exact confidence regions $\Gamma^g_{\alpha,n}$
is that, as a function of $n$, $\Gamma^g_{\alpha,n}$ is not decreasing:
we are not guaranteed to have $\Gamma^g_{\alpha,n+1}\subseteq\Gamma^g_{\alpha,n}$.
This phenomenon of ``losing evidence'' and ways of partially preventing it
are discussed in \cite{Dawid/etal:2011,Shafer/etal:2011} and \cite[Chap.~11]{Shafer/Vovk:2019}.

It is essential to have the martingale test $(S^Q)$ fixed in advance
in order to have valid confidence regions;
on the other hand, confidence regions corresponding to different $g$
are valid simultaneously.

\section{Multiple testing \emph{en masse}}
\label{sec:DM}

In this section we define, for each rejection set $R\subseteq\{1,\dots,K\}$,
a confidence region for the number of justified discoveries in $R$
(i.e., anomalous $k\in R$).
Such confidence regions are often of the form $\{L,\dots,K\}$
for some lower bound $L$
(we only have a lower confidence bound
since $S^{(k)}$ can be arbitrarily close to being a test martingale
without being one).

Given a rejection set $R$, we are interested in the parameter
\begin{equation*}
  g_R(Q)
  :=
  \left|
    R \setminus J_Q
  \right|,
\end{equation*}
which is the number of justified discoveries.
While in this paper we concentrate on the parameter function $g_R$,
this function can be generalized in various directions;
see, e.g., \cite[Remark~6.1]{\eIII}.

The confidence region for $g_R(\P)$ at time $n$ at significance level $\alpha$
consists of $j\in\{1,\dots,K\}$ satisfying $\eP_n(g_R^{-1}(j))<\alpha$,
where the possibility measure $\eP_n(g_R^{-1}(j))$ is
\begin{equation}\label{eq:DM}
  \begin{aligned}
    \eP_n(g_R^{-1}(j))
    &=
    \min_{Q\in\QQQ:g_R(Q)=j}
    S_n^Q
    =
    \inf_{Q\in\QQQ:\left|R\setminus J_Q\right|=j}
    F\left(\left(S^{(k)}_n\right)_{k\in J_Q}\right)\\
    &\ge
    \min_{I\subseteq\{1,\dots,K\}:\left|R\setminus I\right|=j}
    F\left(\left(S^{(i)}_n\right)_{i\in I}\right)
    =:
    D^R(j).
  \end{aligned}
\end{equation}
Replacing $\eP_n(g_R^{-1}(j))$ by $D^R_j$ we also obtain a valid
(perhaps conservative) confidence region.
In the case of the optimal $R:=R_{r,n}$,
we refer to
\[
  D_{r,j}
  :=
  D^{R_{r,n}}(j)
\]
as the \emph{discovery matrix} at time $n$.
It is a lower triangular matrix with $r\in\{1,\dots,K\}$ and $j\in\{0,\dots,r\}$.
In the case $j=r$, the range of $I$ includes the empty set $\emptyset$,
and in this case we set $F$ to 1.

In the computational experiments reported in this paper,
the discovery matrix $D_{r,j}$ is always monotonically decreasing in $j$,
and so
\[
  \eP_n(g_R^{-1}(j))
  =
  \eP_n(g_R^{-1}(\{0,\dots,j\})).
\]
This is essential for the interpretation of our results.
However, in general, the discovery matrix $D_{r,j}$ is not guaranteed
to be decreasing in $j$ \cite{\eIII,\eIV},
and so might need to be regularized
by redefining $D_{r,j}:=\min_{j'\le j}D_{r,j'}$,

\begin{algorithm}[bt]
  \caption{Discovery matrix (lower triangular) $D_{r,j}$}
  \label{alg:DM}
  \begin{algorithmic}[1]
    \Require
      symmetric martingale merging functions $F_k$, $k\in\{1,\dots,K\}$.
    \Require
      decreasing sequence of final martingale values $S^1\ge\dots\ge S^K$.
    \For{$r=1,\dots,K$}
      \For{$j=0,\dots,r$}
        \State $I_{r,j}:=\{j+1,\dots,r\}$
        \State $D_{r,j}:=F((S^i)_{i\in I_{r,j}})$
        \For{$k=r+1,\dots,K$}
          \State $e := F((S^i)_{i\in I_{r,j}\cup\{k,\dots,K\}})$
          \If{$e < D_{r,j}$}
	    \State $D_{r,j} := e$
          \EndIf
        \EndFor
      \EndFor
    \EndFor
  \end{algorithmic}
\end{algorithm}

Algorithm~\ref{alg:DM} implements~\eqref{eq:DM}.
In the case $j=r$ we have $I_{r,j}=\emptyset$,
and as discussed earlier, we set $F(\emptyset):=1$.
This algorithm computes the discovery matrix in time $O(K^4)$
when $F$ is a fixed $U_n$ or a fixed convex mixture of the first few $U_n$;
this follows from $F$ being computable in time $O(K)$,
which in turn follows from, e.g., Newton's identities (see Appendix~\ref{app:Newton}).
It is interesting that for the simulation studies reported in this paper
we do not need more efficient algorithms
such as the $O(K^3)$ algorithm given in \cite{\eIV}
and, in the case of $U_1$, the $O(K^2)$ algorithm given in \cite{\eIII};
computations take at most a couple of minutes on an ordinary laptop.

\begin{figure}
  \begin{center}
    \includegraphics[width=0.49\textwidth]{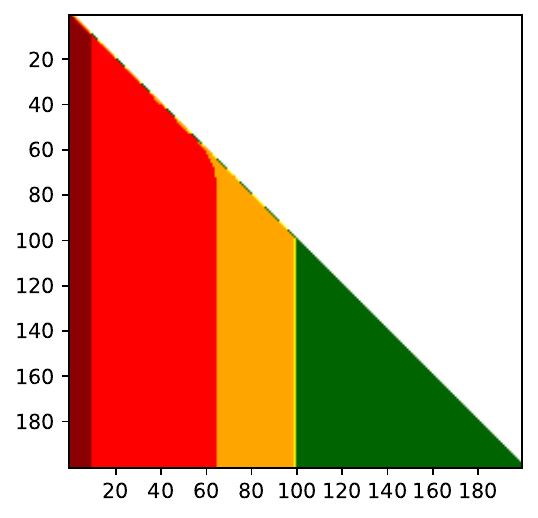}
    \includegraphics[width=0.49\textwidth]{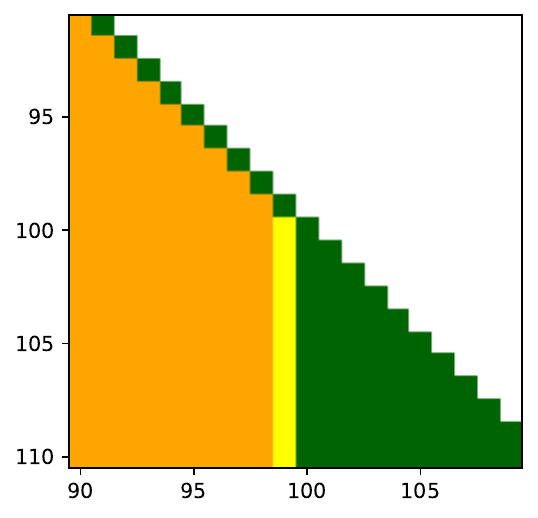}
  \end{center}
  \caption{Left panel:
    Discovery matrix for the mean $U_1$ as martingale merging function.
    The right panel shows its middle portion.}
  \label{fig:DM_mean}
\end{figure}

\begin{figure}
  \begin{center}
    \includegraphics[width=0.49\textwidth]{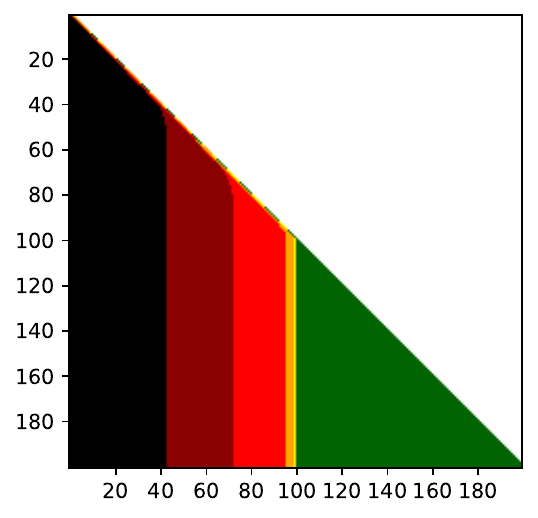}
    \includegraphics[width=0.49\textwidth]{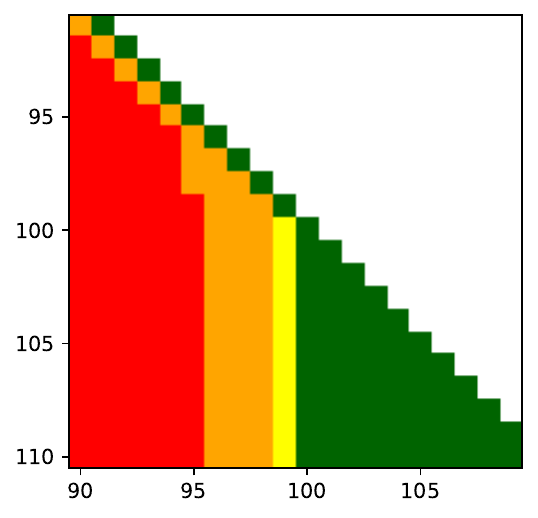}
  \end{center}
  \caption{Left panel:
    Discovery matrix for the mixture $(U_1+U_2)/2$ as martingale merging function.
    The right panel shows its middle portion.}
  \label{fig:DM_mix}
\end{figure}

See Figs~\ref{fig:DM_mean} and~\ref{fig:DM_mix} for discovery matrices
at time 10\,000 for the same simulated data as before.
The colour coding used in both figures
involves much more extreme values of the possibility measures
than the usual scheme used in \cite{\eIII,\eIV}
(the scheme in \cite{\eIII,\eIV} uses the thresholds
proposed by Jeffreys \cite[Appendix~B]{Jeffreys:1961}):
\begin{itemize}
\item
  The final martingale values below $10$ are shown in green.
  For such cells (in row $r$ and column $j$) we have $D_{r,j}<10$,
  and these are exactly the cells for which we do not have strong evidence
  for there being at least $j$ justified discoveries.
\item
  The final martingale values between $10$ and $100$ are shown in yellow.
  These are exactly the cells for which we have strong but not decisive  evidence
  for there being at least $j$ justified discoveries
  ($10\le D_{r,j}<100$).
\item
  The final martingale values between $100$ and $10^8$ are shown in orange.
  For these cells (and for the cells in the darker colours)
  we have decisive  evidence for there being at least $j$ justified discoveries.
\item
  The final martingale values between $10^8$ and $10^{14}$
  are shown in red.
\item
  The final martingale values between $10^{14}$ and $10^{20}$
  are shown in dark red.
\item
  The final martingale values above $10^{20}$ are shown in black.
\end{itemize}

The \emph{diagonal} of a discovery matrix consists of the cells $D_{r,r-1}$,
the justification for the offset of 1 being that $j$ in $D_{r,j}$ starts from 0 rather than 1.
Correspondingly, the \emph{subdiagonal} consists of $D_{r,r-2}$,
and the \emph{superdiagonal} consists of $D_{r,r}$.

The diagonal, subdiagonal, and superdiagonal can be clearly seen
in the right panels of the two figures.
The diagonal can be traced starting from the top left corner of either bounding box.
In both figures the diagonal (when moving south-east)
is first orange, then yellow (just one cell), and then green.
In Fig.~\ref{fig:DM_mean} the subdiagonal
is also first orange, then yellow (just one cell), and then green,
but in Fig.~\ref{fig:DM_mix} the subdiagonal is first red
and only later becomes orange.
The superdiagonal is green in both figures.

The corresponding confidence intervals
(i.e., confidence regions that happen to be intervals)
can be read off the two figures.
For example, for each row $r$,
the non-green cells represent the confidence interval
for the number of justified discoveries among the top $r$ martingale values
at significance level 10.
We can see that for $r=100$ the confidence interval is $\{100\}$;
it is degenerate and contains only one value:
we are predicting that all 100 null hypotheses with the largest final martingale values
are justified (and \emph{a fortiori} true) discoveries,
and we have strong evidence for that.
On the other hand, the green superdiagonal entry $D_{100,100}$ is very small
(it is, approximately, $1.13\times10^{-20}$).
If we raise the significance level to 100
(Jeffreys's threshold for decisive evidence),
the confidence interval widens to $\{99,100\}$.
And when we raise it further to the huge value of $10^8$,
the confidence interval (given by the non-red entries in the right panel)
widens to $\{96,97,98,99,100\}$.

Figures~\ref{fig:DM_mean} and~\ref{fig:DM_mix}
suggest that discovery matrices are monotonically decreasing
in the eastern and south-eastern directions
and monotonically increasing in the southern direction.
These properties of monotonicity
(except for the monotonicity in $j$ discussed earlier)
are stated and proved in \cite{\eIII} and \cite{\eIV}.

Figures~\ref{fig:plot_diagonal} and~\ref{fig:plot_subdiagonal}
show the evolution of various entries of discovery matrices
such as those in Figs~\ref{fig:DM_mean} and~\ref{fig:DM_mix}
over time.
The green lines in both panels of Fig.~\ref{fig:plot_diagonal}
show the evolution of the diagonal entry $D_{100,99}$
over the 10\,000 observations.
The orange and blue lines in Fig.~\ref{fig:plot_diagonal}
show the evolution of the entries $D_{99,98}$ and $D_{98,97}$, respectively.
All these entries lie on the diagonal
\[
  d_{r,10000}
  :=
  D_{r,r-1}
\]
of the discovery matrix at time 10\,000.
We talked about family-wise validity in Sect.~\ref{sec:diagonal}
since $D_{r,r-1}$ is the largest significance level
at which the confidence interval is a one-element set,
namely $\{r\}$.

The right panel of Fig.~\ref{fig:plot_diagonal} also shows,
as red line, the evolution of the diagonal entry $D_{101,100}$.
This line shows that the green entry $D_{101,100}$ in Fig.~\ref{fig:DM_mean}
is very small;
in numbers, the final value of the red line in Fig.~\ref{fig:plot_diagonal}
(i.e., $D_{101,100}$ in Fig.~\ref{fig:DM_mean})
is, approximately, $7.80\times10^{-7}$).
The value of $D_{101,100}$ in Fig.~\ref{fig:DM_mix} is even smaller.

The green line in Fig.~\ref{fig:plot_subdiagonal}
shows the evolution of the entry $D_{100,98}$,
which determines the significance levels at which
the confidence interval for the number of justified discoveries
is $\{99,100\}$ (allowing one unjustified discovery).
Its final value corresponds to the entry $D_{100,98}$ in Fig.~\ref{fig:DM_mix},
and the two numbers have the same order of magnitude
(they are, however, different
because Figs~\ref{fig:plot_subdiagonal} and~\ref{fig:DM_mix}
use different martingale merging functions, $U_2$ vs $(U_1+U_2)/2$).
The orange and blue lines in Fig.~\ref{fig:plot_subdiagonal}
are interpreted in the same way;
they correspond to the entries $D_{99,97}$ and $D_{98,96}$, respectively,
of Fig.~\ref{fig:DM_mix}.
The red line in Fig.~\ref{fig:plot_subdiagonal}, however,
disagrees sharply with the entry $D_{101,99}$ of Fig.~\ref{fig:DM_mix},
because the martingale merging function $(U_1+U_2)/2$ used in Fig.~\ref{fig:DM_mix}
has $U_1$ as its component.

\section{Conclusion}

These are some possible directions of further research:
\begin{itemize}
\item
  The motivation behind this paper is coming from game-theoretic probability and statistics,
  but its mathematical setting is that of measure-theoretic probability.
  Replacing measure-theoretic probability by purely game-theoretic probability
  (as developed in \cite{Shafer/Vovk:2019})
  would simplify the exposition and lead to more natural and general definitions.
\item
  This paper concentrates on simulation studies.
  It would be interesting to conduct empirical studies
  on benchmark or real-world datasets,
  for example ones collected in the course of statistical meta-analyses.
\item
  The experimental results of Sect.~\ref{sec:DM} establish confidence regions
  for the numbers of true discoveries,
  which can be restated as results about the false discovery proportions, FDP.
  Are there any interesting theoretical results in this context
  about false discovery rates, FDR
  (as in \cite{Benjamini/Hochberg:1995} in the case of p-values
  and \cite{Wang/Ramdas:2022} in the case of e-values)?
\item
  This paper concentrates on the closed setting
  (when the number of null hypotheses $K$ is given in advance).
  The open setting, where new hypotheses may appear at any moment,
  may be even more interesting.
  In this case we need, of course, to break the symmetry between the null hypotheses:
  there is no uniform probability measure on $\{1,2,\dots\}$.
\end{itemize}

\subsection*{Acknowledgments}

Many thanks to Jean Gallier for his advice on literature
and for correcting the statement of Lemma~4.1.3 in \cite{Gallier:2024}.
My research has been partially supported by Mitie.

\appendix

\section{Proofs}
\label{app:proofs}

In this appendix I will prove Propositions~\ref{prop:main} and~\ref{prop:symmetric}
(and state and prove a new Proposition~\ref{prop:domination}).
A \emph{multiaffine function} is a multivariate function
that is affine in each of its arguments.
(So that multiaffine polynomials are multiaffine functions,
and in Sect.~\ref{subsec:proof-main} we will see that these two notions are equivalent.)
The proofs will follow from the following lemma.

\begin{lemma}\label{lem:multiaffine}
  A martingale merging function must be multiaffine.
\end{lemma}

\begin{proof}
  Let $F$ be a martingale merging function.
  We are required to prove
  \begin{multline}\label{eq:goal}
    F(s_1,\dots,s_{k-1},\alpha s'_k+(1-\alpha)s''_k,s_{k+1},\dots,s_K)\\
    =
    \alpha
    F(s_1,\dots,s_{k-1},s'_k,s_{k+1},\dots,s_K)\\
    +
    (1-\alpha)
    F(s_1,\dots,s_{k-1},s''_k,s_{k+1},\dots,s_K),
  \end{multline}
  where $\alpha\in(0,1)$.
  Let us fix $s_1,\dots,s_{k-1},s'_k,s''_k,s_{k+1},\dots,s_K$, and $\alpha$.
  Consider the sample space $\Omega:=\{0,1\}^N$ with the natural filtration
  and a positive probability measure $\P$
  (i.e., $\P(E)>0$ for any $E\ne\emptyset$).
  Suppose that the set of sample points where
  \begin{multline*}
    S^{(1)}_{N-1}=s_1,\dots,
    S^{(k-1)}_{N-1}=s_{k-1},
    S^{(k)}_{N-1}=\alpha s'_k+(1-\alpha)s''_k,\\
    S^{(k+1)}_{N-1}=s_{k+1},\dots,
    S^{(K)}_{N-1}=s_K
  \end{multline*}
  for some uncorrelated test martingales $S^{(1)},\dots,S^{(K)}$
  is non-empty,
  and let $(\omega_1,\dots,\omega_N)\in\Omega$ be such a sample point.
  Suppose that in our probability space we have the branching probability
  \[
    \frac
      {\P(\{(\omega_1,\dots,\omega_{N-1},1)\})}
      {\P(\{(\omega_1,\dots,\omega_{N-1},0),(\omega_1,\dots,\omega_{N-1},1)\})}
    =
    \alpha
  \]
  and that the martingale $S^{(k)}$ satisfies
  \begin{align*}
    S^{(k)}_N((\omega_1,\dots,\omega_{N-1},1))
    &=
    s'_k\\
    S^{(k)}_N((\omega_1,\dots,\omega_{N-1},0))
    &=
    s''_k.
  \end{align*}
  The existence of such a probability space
  and uncorrelated test martingales $S^{(1)},\dots,S^{(K)}$
  is obvious.
  Since
  \[
    T_n
    :=
    F(S^{(1)}_n,\dots,S^{(K)}_n)
  \]
  is a test martingale,
  we have
  \[
    T_{N-1}((\omega_1,\dots,\omega_{N}))
    =
    \alpha
    T_N((\omega_1,\dots,\omega_{N-1},1))
    +
    (1-\alpha)
    T_N((\omega_1,\dots,\omega_{N-1},0)),
  \]
  which is equivalent to \eqref{eq:goal}.
\end{proof}

\subsection{Proof of Proposition~\ref{prop:main}}
\label{subsec:proof-main}

To show that a martingale merging function is a multiaffine polynomial,
we combine Lemma~\ref{lem:multiaffine}
with Lemma~4.1.3 in \cite{Gallier:2024}
(whose proof relies on Cartan's method of successive differences
\cite[Sect.~6.3]{Cartan:1967}).
According to \cite[Lemma~4.1.3]{Gallier:2024},
a multiaffine function $f$ of $K$ arguments has the form
\begin{equation*}
  f(s_1,\dots,s_K)
  =
  f(0,\dots,0)
  +
  \sum_{\substack{n\in\{1,\dots,K\}\\
    \{1\le k_1\le\dots\le k_n\le K\}}}
  f_{k_1,\dots,k_n}(s_{k_1},\dots,s_{k_n}),
\end{equation*}
where $f_{k_1,\dots,k_n}$ are multilinear functions
(i.e., functions linear in each argument).
It remains to notice that 
\[
  f_{k_1,\dots,k_n}(s_{k_1},\dots,s_{k_n})
  =
  c
  s_{k_1} \dots s_{k_n}
\]
for some constant $c$;
indeed,
\begin{align*}
  f_{k_1,\dots,k_n}(s_{k_1},\dots,s_{k_n})
  &=
  s_{k_1}
  f_{k_1,\dots,k_n}(1,s_{k_2},\dots,s_{k_n})\\
  &=
  s_{k_1} s_{k_2}
  f_{k_1,\dots,k_n}(1,1,s_{k_3},\dots,s_{k_n})=\dots\\
  &=
  s_{k_1}\dots s_{k_n}
  f_{k_1,\dots,k_n}(1,\dots,1).
\end{align*}

Let us now check that a martingale merging function $F$
is a positive multiaffine polynomial.
Suppose there is a negative coefficient in front of one or more of its monomials.
Choose and fix a monomial with a negative coefficient.
Set all variables that do not occur in this monomial to zero.
Set each of the variables that do occur in this monomial to $C$ and let $C\to\infty$.
For a large enough $C$, the value of the polynomial
(the value being a univariate polynomial in $C$
with a negative leading coefficient)
will become negative, which is impossible.

There is a minor gap in our derivation of Proposition~\ref{prop:main}
from \cite[Lemma~4.1.3]{Gallier:2024}:
the latter assumes that the multiaffine function $f$
is defined on an affine space
whereas in our context $f$ is defined on $[0,\infty)^K$.
Let us check that every affine $f:[0,\infty)^K\to\R$
can be extended to an affine $f':\R^K\to\R$.
We proceed by induction and show that if $f(x_1,\dots,x_k,\dots,x_K)$
is an affine function with $x_k$ ranging over $[0,\infty)$
we can extend it to an affine function with $x_k$ ranging over $\R$
(with the ranges of the other arguments of $f$ unchanged).
Without loss of generality, let $k:=1$.
We extend $f$ to $f'$ by the affinity in $x_1$:
for any $x_1<0$,
\begin{equation}\label{eq:definition}
  f'(x_1,x_2,\dots)
  :=
  x_1
  f'(1,x_2,\dots)
  +
  (1-x_1)
  f'(0,x_2,\dots).
\end{equation}
We only need to check that $f'$ is multiaffine.
The affinity in $x_1$ holds by construction,
so we only need to check that $f'$ is affine in $x_k$ for $k\ne1$.
Without loss of generality, let $k:=2$.
Since the arguments $x_3,\dots,x_K$ of $f$ and $f'$ are kept fixed,
we will ignore them.
Our goal is to show that
\begin{equation}\label{eq:goal-1}
  f'
  \left(
    x_1,\alpha x'_2+(1-\alpha)x''_2
  \right)
  =
  \alpha
  f'
  \left(
    x_1,x'_2
  \right)
  +
  (1-\alpha)
  f'
  \left(
    x_1,x''_2
  \right)
\end{equation}
for $x_1<0$.
By the definition \eqref{eq:definition},
the equality \eqref{eq:goal-1} is equivalent to
\begin{multline}\label{eq:goal-2}
  x_1
  f
  \left(
    1,\alpha x'_2+(1-\alpha)x''_2
  \right)
  +
  (1-x_1)
  f
  \left(
    0,\alpha x'_2+(1-\alpha)x''_2
  \right)\\
  =
  \alpha
  x_1
  f
  \left(
    1,x'_2
  \right)
  +
  \alpha
  (1-x_1)
  f
  \left(
    0,x'_2
  \right)\\
  +
  (1-\alpha)
  x_1
  f
  \left(
    1,x''_2
  \right)
  +
  (1-\alpha)
  (1-x_1)
  f
  \left(
    0,x''_2
  \right).
\end{multline}
It remains to notice that \eqref{eq:goal-2} can be derived
as linear combination of
\begin{equation}\label{eq:component-1}
  f
  \left(
    1,\alpha x'_2+(1-\alpha)x''_2
  \right)
  =
  \alpha
  f
  \left(
    1,x'_2
  \right)
  +
  (1-\alpha)
  f
  \left(
    1,x''_2
  \right)
\end{equation}
and
\begin{equation}\label{eq:component-2}
  f
  \left(
    0,\alpha x'_2+(1-\alpha)x''_2
  \right)
  =
  \alpha
  f
  \left(
    0,x'_2
  \right)
  +
  (1-\alpha)
  f
  \left(
    0,x''_2
  \right)
\end{equation}
(with the coefficients $x_1$ for \eqref{eq:component-1}
and $1-x_1$ for \eqref{eq:component-2}).

\subsection{Proof of Proposition~\ref{prop:symmetric}}
\label{subsec:proof-symmetric}

We proceed as in Sect.~\ref{subsec:proof-main}
replacing Lemma~4.1.3 in \cite{Gallier:2024} by Lemma~4.1.4.
Alternatively, we could have derived Proposition~\ref{prop:symmetric}
from Proposition~\ref{prop:main}.

\subsection{Comparisons with merging independent and sequential e-values}
\label{subsec:merging}

In this subsection we will discuss ie-merging and se-merging functions,
to be defined momentarily;
for a further discussion of these functions see, e.g., \cite{\eV}.

\begin{figure}
  \begin{center}
    \includegraphics[trim={7.5cm 20cm 0 4.5cm},clip,width=\textwidth]{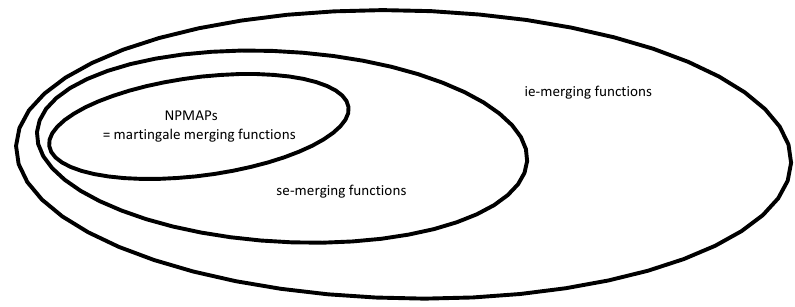}
  \end{center}
  \caption{Three families of merging functions;
    all inclusions in this Euler diagram are strict.}
  \label{fig:structure}
\end{figure}

Suppose $E_1,\dots,E_K$ are admissible independent e-variables
(i.e., nonnegative random variables that are independent and satisfy $\E(E_k)=1$, $k=1,\dots,K$)
or admissible sequential e-variables
(i.e., nonnegative, adapted, and satisfying $\E(E_k\mid\FFF_k-1)=1$, $k=1,\dots,K$).
Then
\[
  S^{(k)}_n
  :=
  \begin{cases}
    1 & \text{if $n<k$}\\
    E_k & \text{if $n\ge k$}
  \end{cases}
\]
are uncorrelated test martingales with final values $E_1,\dots,E_k$.
Therefore, any normalized positive multiaffine polynomial (NPMAP)
is an ie-merging function,
in the sense of mapping any (admissible) independent e-variables to an e-variable
(i.e., nonnegative random variable $E$ satisfying $\E(E)\le1$);
moreover, any NPMAP is an se-merging function,
in the sense of mapping any (admissible) sequential e-variables to an e-variable.
This gives us the structure shown in Fig.~\ref{fig:structure}:
it is obvious that every se-merging function is an ie-merging function.

Let us check that both inclusions in the Euler diagram
shown in Fig.~\ref{fig:structure} are strict.
The outer inclusion is strict since the function
\begin{equation}\label{eq:f}
  f(e_1,e_2)
  :=
  \frac12
  \left(
    \frac{e_1}{1 + e_1}
    +
    \frac{e_2}{1 + e_2}
  \right)
  \left(
    1 + e_1 e_2
  \right)
\end{equation}
is an admissible ie-merging function \cite[Remark~4.3]{\eII}
while it is not an se-merging function \cite[Example~2]{\eV}.
To see that the inner inclusion is strict,
notice that
\[
  (e_1,\dots,e_K)
  \mapsto
  1 + g(e_1,\dots,e_{K-1})
  (e_K-1)
\]
is an se-merging function for any function $g$ taking values in $[0,1]$,
even highly non-linear one, such as $g(e_1,\dots,e_{K-1}):=(\sin e_1+1)/2$.

\begin{figure}
  \begin{center}
    \includegraphics[trim={7.5cm 20cm 0 4.5cm},clip,width=\textwidth]{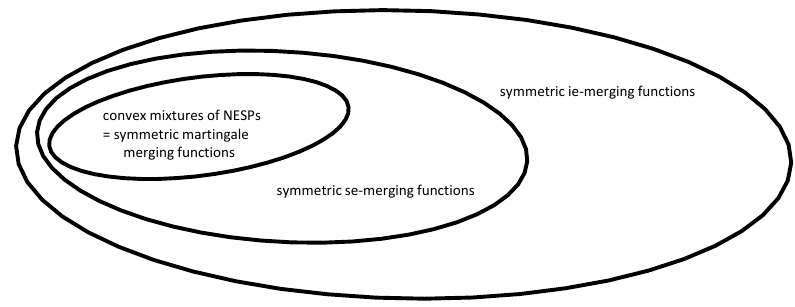}
  \end{center}
  \caption{Three families of symmetric merging functions.
    All inclusions in this Euler diagram are strict,
    but each symmetric se-merging function
    is dominated by a convex mixture of NESPs.}
  \label{fig:structure-symmetric}
\end{figure}

Specializing Fig.~\ref{fig:structure} to symmetric merging functions
we obtain Fig.~\ref{fig:structure-symmetric}.
The function \eqref{eq:f} is symmetric
and so can also serve as an example demonstrating
that the outer inclusion in Fig.~\ref{fig:structure-symmetric} is strict.
On the other hand, the following proposition shows
that the inner inclusion in Fig.~\ref{fig:structure-symmetric} is strict
in an uninteresting way.

\begin{proposition}\label{prop:domination}
  Every symmetric se-merging function
  is dominated by a convex mixture of NESPs.
\end{proposition}

\begin{proof}
  Theorem 1 in \cite{\eV} says that any se-merging function
  is dominated by a martingale merging function
  (where ``martingale merging function'' is used in a sense
  different from this paper;
  this proof uses ``martingale merging function'' in the sense of \cite{\eV}).
  By definition, a martingale merging function is affine in its last argument.
  By symmetry, it is affine in each argument.
  It remains to follow the reasoning
  of Sects~\ref{subsec:proof-main} and~\ref{subsec:proof-symmetric}.
\end{proof}

Proposition~\ref{prop:domination} appears to be less interesting
than Propositions~\ref{prop:main} and~\ref{prop:symmetric}:
the family of symmetric se-merging functions is not as natural
as the other two symmetric families in Fig.~\ref{fig:structure-symmetric}.

\section{Computing NESPs}
\label{app:Newton}

In this appendix we will discuss how to compute
the NESP $U_n=U_n(s_1,\dots,s_K)$ for a fixed $n$ efficiently,
namely, in time $O(K)$.
We will use the fact that the polynomials
$p_n:=s_1^n+\dots+s_K^n$ can be computed in time $O(K)$,
and so one way to compute $U_n$ efficiently is to express them via $p_n$.

These are the efficient representations for the first few NESPs:
\begin{align*}
  U_1(s_1,\dots,s_K)
  &=
  \frac{1}{K}
  (s_1+\dots+s_K)\\
  U_2(s_1,\dots,s_K)
  &=
  \frac{1}{K(K-1)}
  \left(
    (s_1+\dots+s_K)^2
    -
    (s_1^2+\dots+s_K^2)
  \right)\\
  U_3(s_1,\dots,s_K)
  &=
  \frac{1}{K(K-1)(K-2)}
  \bigl(
    (s_1+\dots+s_K)^3\\
    &\quad-
    3
    (s_1^2+\dots+s_K^2)
    (s_1+\dots+s_K)
    +
    2
    (s_1^3+\dots+s_K^3)
  \bigr)\\
  U_4(s_1,\dots,s_K)
  &=
  \frac{1}{K(K-1)(K-2)(K-3)}
  \bigl(
    (s_1+\dots+s_K)^4\\
    &\quad-
    6
    (s_1^2+\dots+s_K^2)
    (s_1+\dots+s_K)^2\\
    &\quad+
    8
    (s_1^3+\dots+s_K^3)
    (s_1+\dots+s_K)\\
    &\quad+
    3
    (s_1^2+\dots+s_K^2)^2
    -
    6
    (s_1^4+\dots+s_K^4)
  \bigr).
\end{align*}
It is clear that such a representation exists for any fixed $n$,
and it allows us to compute $U_n(s_1,\dots,s_K)$ in time $O(K)$.
In terms of \emph{Bell polynomials}
\begin{equation*}
  B_n(x_1,\dots,x_n)
  :=
  n!
  \sum_{\substack{(j_1,\dots,j_n)\in\N^n:\\j_1+2j_2+\dots+n j_n=n}}
  \prod_{i=1}^n
  \frac{x_i^{j_i}}{(i!)^{j_i}j_i!},
\end{equation*}
where $\N:=\{0,1,\dots\}$ is the set of natural numbers,
the general expression is
\begin{equation*}
  U_n(s_1,\dots,s_K)
  =
  \frac{(K-n)!}{K!}
  B_{n}(p_1, -p_2, 2! p_3, -3! p_4, \dots, (-1)^{n-1}(n-1)! p_n),
\end{equation*}
where $p_n:=s_1^n+\dots+s_K^n$.

A less straightforward way of computing the elementary symmetric polynomials
(and therefore, $U_n$)
via $p_1,p_2,\dots$ in time $O(K)$
would be to use recursion and Newton's identities
(see, e.g., \cite[Lemma~4 of Chap.~4]{Bourbaki:algebra-II}
or \cite[Theorem~4.5.5]{Nicholson:1999}).
\end{document}